\documentstyle{llncs}

\def\ra{\rightarrow}
\def\la{\leftarrow}

\newbox\tempa
\newbox\tempb
\newdimen\tempc
\def\mud#1{\hfil $\displaystyle{\mathstrut #1}$\hfil}
\def\rig#1{\hfil $\displaystyle{#1}$}
\def\irulehelp#1#2#3{\setbox\tempa=\hbox{$\displaystyle{\mathstrut #2}$}%
		        \setbox\tempb=\vbox{\halign{##\cr
	\mud{#1}\cr
	\noalign{\vskip\the\lineskip}%
	\noalign{\hrule height 0pt}%
	\rig{\vbox to 0pt{\vss\hbox to 0pt{${\; #3}$\hss}\vss}}\cr
	\noalign{\hrule}%
	\noalign{\vskip\the\lineskip}%
	\mud{\copy\tempa}\cr}}%
		      \tempc=\wd\tempb
		      \advance\tempc by \wd\tempa
		      \divide\tempc by 2 }
\def\irule#1#2#3{{\irulehelp{#1}{#2}{#3}%
		     \hbox to \wd\tempa{\hss \box\tempb \hss}}}

\begin{document}

\title{The Undecidability of Typability in the Lambda-Pi-Calculus}

\author{Gilles Dowek}

\institute{School of Computer Science, Carnegie Mellon University\\
Pittsburgh,  PA 15213-3890, U.S.A.}

\maketitle

\begin{abstract}
The set of pure terms which are typable in the $\lambda \Pi$-calculus in 
a given context is not recursive. So there is no general type inference 
algorithm for the programming language Elf and, in some cases, some type 
information has to be mentioned by the programmer.
\end{abstract}

\section*{Introduction}
\thispagestyle{empty}
The programming language Elf \cite{Pfenning90} is an extension of 
$\lambda$-Prolog in which
the clauses are expressed in a $\lambda$-calculus with dependent types 
($\lambda \Pi$-calculus \cite{HHP}). Since this calculus verifies the 
propositions-as-types principle, a proof of a proposition is merely a term of 
the calculus. Using this property of the $\lambda \Pi$-calculus, the programmer
can either express a proposition and let the machine search for a proof of 
this proposition (as in usual logic programming) or express both a proposition
and its proof and let the machine check that this proof is correct (as in 
proof-verification systems). Thus Elf can be used both to express logic 
programs and to reason about of their properties.

A {\it type inference algorithm} for a given language is an 
algorithm which assigns a type to each variable of a program. Thus, when such 
an algorithm exists, the types of the variables do not need to be mentioned by
the programmer. As an example, a type inference algorithm for the language ML
is given in \cite{DamMil}.

We show here that the set of pure terms which are typable in the
$\lambda \Pi$-calculus in a given context is not recursive. So there is no 
general type inference algorithm for the language Elf and, in some cases, some
type information has to be mentioned by the programmer. 

As already remarked in \cite{DamMil}, typing a term requires the solution of a
unification problem. Typing a term in the simply typed $\lambda$-calculus (and
in ML) requires the solution of a first order unification problem and thus 
typability is decidable in these languages.

Typing a term in the $\lambda \Pi$-calculus requires the solution of a
unification problem which is also formulated in the $\lambda \Pi$-calculus. 
Unification in the $\lambda \Pi$-calculus has been shown to be undecidable 
(third order unification in \cite{Huet73}, then second order unification 
in \cite{Goldfarb} and third order pattern matching in \cite{cras}), i.e. 
there is no algorithm that decides if such a unification problem has a 
solution. But in order to prove the undecidability of typability in the 
$\lambda \Pi$-calculus we need to prove that there is no algorithm that 
decides if a unification problem {\it produced by a typing problem} has a 
solution. Unification problems produced by typing problems are very restricted
and the undecidability proofs of unification have to be adapted to this class 
of problems. We show here that the proof of \cite{Huet73} can easily be 
adapted. 

\section{The Lambda-Pi-Calculus}

We follow \cite{Barendregt} for a presentation of the $\lambda \Pi$-calculus. 
The set of {\it terms} is inductively defined by
$$T~::=~Type~|~Kind~|~x~|~(T~T)~|~\lambda x:T.T~|~\Pi x:T.T$$

In this note, we ignore variable renaming problems. A rigorous presentation 
would use de Bruijn indices. The terms $Type$ and $Kind$ are called 
{\it sorts}, the terms $x$ {\it variables}, the terms $(t~t')$ 
{\it applications}, the terms $\lambda x:t.t'$ {\it abstractions} and the 
terms $\Pi x:t.t'$ {\it products}. The notation $t \ra t'$ is used for 
$\Pi x:t.t'$ when $x$ has no free occurrence in $t'$.

Let $t$ and $t'$ be terms and $x$ a variable. We write $t[x \la t']$ for the 
term obtained by substituting $t'$ for $x$ in $t$. We write $t \cong t'$ 
when the terms
$t$ and $t'$ are $\beta$-equivalent ($\beta \eta$-equivalence can also be 
considered without affecting the proof given here). 

A {\it context} is a list of pairs $<x,T>$ (written $x:T$) where $x$
is a variable and $T$ a term. 

We define inductively two judgements: {\it $\Gamma$ is well-formed} and 
{\it $t$ has type $T$ in $\Gamma$} ($\Gamma \vdash t:T$) where $\Gamma$ is
a context and $t$ and $T$ are terms.
$$\irule{}
        {[~]~\mbox{well-formed}}
        {}$$
$$\irule{\Gamma \vdash T:s}
        {\Gamma[x:T]~\mbox{well-formed}}
        {s \in \{Type,Kind\}}$$
$$\irule{\Gamma~\mbox{well-formed}}
        {\Gamma \vdash Type:Kind}
        {}$$
$$\irule{\Gamma~\mbox{well-formed}~~x:T \in \Gamma} 
        {\Gamma \vdash x:T}
        {}$$
$$\irule{\Gamma \vdash T:Type~~\Gamma[x:T] \vdash T':s}
        {\Gamma \vdash \Pi x:T.T':s}
        {s \in \{Type,Kind\}}$$
$$\irule{\Gamma \vdash \Pi x:T.T':s~~\Gamma[x:T] \vdash t:T'}
        {\Gamma \vdash \lambda x:T.t:\Pi x:T.T'}
        {s \in \{Type,Kind\}}$$
$$\irule{\Gamma \vdash t:\Pi x:T.T'~~\Gamma \vdash t':T}
        {\Gamma \vdash (t~t'):T'[x \la t']}
        {}$$
$$\irule
     {\Gamma \vdash T:s~~\Gamma \vdash T':s~~\Gamma \vdash t:T~~T \cong T'}
        {\Gamma \vdash t:T'}
        {s \in \{Type,Kind\}}$$

A term $t$ is said to be {\it well-typed} in a context $\Gamma$ if there exists
a term $T$ such that $\Gamma \vdash t:T$.

The reduction relation on well-typed terms is strongly normalizable and 
confluent. Thus each well-typed term has a unique normal form and two terms 
are equivalent if they have the same normal form \cite{HHP} 
(\cite{Geuvers} \cite{Salvesen} \cite{Coquand91} if $\beta \eta$-equivalence 
is considered).

A term $t$ well-typed in a context $\Gamma$ has a unique type modulo 
equivalence.

A normal term $t$ well-typed in a context $\Gamma$ has either the form 
$$t = \lambda x_{1}:T_{1}....\lambda x_{n}:T_{n}.(x~c_{1}~...~c_{n})$$
where $x$ is a variable or a sort or 
$$t = \lambda x_{1}:T_{1}....\lambda x_{n}:T_{n}.\Pi x:P.Q$$
The {\it head symbol} of $t$ is $x$ is the first case and, by convention, the
symbol $\Pi$ in the second. The {\it top variables} of $t$ are the variables 
$x_{1}, ..., x_{n}$. 

\section{Typability in the Lambda-Pi-Calculus}

\begin{definition}
A term $t$ of type $T$ in a context $\Gamma$ is said to be 
an {\it object} in $\Gamma$ if $\Gamma \vdash T:Type$. 
\end{definition}

\begin{proposition}
If a term $t$ is an object in a context $\Gamma$ then it is either a 
variable, an application or an abstraction. If it is an application 
$t = (u~v)$
then both terms $u$ and $v$ are objects in $\Gamma$, if it is an abstraction
$t = \lambda x:U.u$ then the term $u$ is an object in the context 
$\Gamma [x:U]$. 
\end{proposition}

\begin{definition}
The set of {\it pure terms} is inductively defined by
$$T~::=~x~|~(T~T)~|~\lambda x.T$$
\end{definition}

\begin{definition}
Let $t$ be an object in a context $\Gamma$, the {\it content} of $t$ ($|t|$) is
the pure term defined by induction over the structure of $t$ by\\
$\bullet$
$|x| = x$,\\
$\bullet$
$|(t~t')| = (|t|~|t'|)$,\\
$\bullet$
$|\lambda x:U.t| = \lambda x.|t|$.

A pure term $t$ is said to be {\it typable} in a context $\Gamma$ if there 
exists a term $t'$ well-typed in an extension $\Gamma \Delta$ of $\Gamma$ such
that $t'$ is an object in $\Gamma \Delta$ and $t = |t'|$.
\end{definition}

\begin{remark}
Typing a pure term in a context $\Gamma$ is assigning a type to bound 
variables and to some of the free variables, while the type of the other free 
variables is given in the context $\Gamma$. When the context $\Gamma$ is empty,
then typing a term in $\Gamma$ is assigning a type to both bound and free 
variables.
\end{remark}

\begin{proposition}
Typability in the empty context is decidable in the $\lambda \Pi$-calculus. 
\end{proposition}
\begin{proof}
Pure terms typable in the empty context in the $\lambda \Pi$-calculus 
and in the simply typed $\lambda$-calculus are the same \cite{HHP} and 
typability is decidable in simply typed $\lambda$-calculus \cite{DamMil}. 
\end{proof}

\section{Post Correspondence Problem}

\begin{definition} Post Correspondence Problem\\
A {\it Post correspondence problem} is a finite set of pairs of words over the
two letters alphabet $\{A,B\}$ : 
$\{<\varphi_{1}, \psi_{1}>, ..., <\varphi_{n}, \psi_{n}>\}$. 
A {\it solution} to such a problem is a non empty sequence of integers
$i_{1}, ..., i_{p}$ such that 
$$\varphi_{i_{1}} ... \varphi_{i_{p}} = \psi_{i_{1}} ... \psi_{i_{p}}$$
\end{definition}

\begin{theorem} {} (Post \cite{Post})
It is undecidable whether or not a Post problem has a solution.
\end{theorem}

\section{Undecidability of Typability in the Lambda-Pi-Calculus}

Let us consider the context

\noindent $\Gamma = [T:Type;a:T \ra T;b:T \ra T;c:T; d:T; P:T \ra Type;$

\hfill $F:\Pi x:T.((P~x) \ra T)]$

\begin{definition} {} (Huet \cite{Huet73})
Let $\varphi$ be a word in the two 
letters alphabet $\{A,B\}$, we define by induction on the length of $\varphi$ 
the term $\hat{\varphi}$ well-typed in $\Gamma$ and the pure term 
$\tilde{\varphi}$ as follows 
$$\hat{\varepsilon} = \lambda y:T.y$$
$$\hat{A \varphi} = \lambda y:T.(a~(\hat{\varphi}~y))$$
$$\hat{B \varphi} = \lambda y:T.(b~(\hat{\varphi}~y))$$

$$\tilde{\varphi} = |\hat{\varphi}|$$
\end{definition}

\begin{proposition}
Let 
$\{<\varphi_{1}, \psi_{1}>, ..., <\varphi_{n}, \psi_{n}>\}$ be a Post problem, 
the non empty sequence $i_{1}, ..., i_{p}$ is a solution to this problem if and
only if 
$$(\hat{\varphi_{i_{1}}}~(...(\hat{\varphi_{i_{p}}}~c)...)) \cong 
                          (\hat{\psi_{i_{1}}}~(...(\hat{\psi_{i_{p}}}~c)...))$$
\end{proposition}

\begin{proposition}
If $g$ is a term such that the term $(g~a~...~a)$ ($n$ symbols $a$) is 
well-typed and is an object in an extension $\Gamma \Delta$ of $\Gamma$ then 
the term $g$ is well-typed in the context $\Gamma \Delta$ and its type is 
equivalent to the term 
$$\Pi x_{1}:T \ra T....\Pi x_{n}:T \ra T.(\beta~x_{1}~...~x_{n})$$
for some term $\beta$ of type $(T \ra T) \ra ... \ra (T \ra T) \ra Type$ in 
the context $\Gamma \Delta$.
\end{proposition}
\begin{proof} 
By induction on $n$. 
\end{proof}

\begin{proposition}
Let $t, u_{1}, ..., u_{n}, v$ be normal terms such that 
$(t~u_{1}~...~u_{n})$ is a well-typed term and its normal form is $v$. The 
head symbol of the $t$ is either the head symbol of $v$ or a top variable of 
$t$.
\end{proposition}
\begin{proof}
Let $x$ be the head symbol of $t$. If $x$ is not a top variable of
$t$ then the head symbol of the normal form of $(t~u_{1}~...~u_{n})$ is also 
$x$, so $x$ is the head symbol of $v$. 
\end{proof}

\begin{proposition}
Let $t$ be a normal term of type $(T \ra T) \ra ... \ra (T \ra T) \ra T$
in the context $\Gamma$
such that the normal form of $(t~\lambda y:T.y~...~\lambda y:T.y)$ is equal to
$c$. Then the term $t$ has the form
$$t = \lambda x_{1}:T \ra T....\lambda x_{n}:T \ra T.
(x_{i_{1}}~(...(x_{i_{p}}~c)...))$$
for some sequence $i_{1}, ..., i_{p}$. 
\end{proposition}
\begin{proof}
By induction on the number of variable occurrences in $t$.
\end{proof}

\begin{theorem}
It is undecidable whether or not a pure term is typable in a 
given context. 
\end{theorem}
\begin{proof}
Consider a Post problem 
$\{<\varphi_{1}, \psi_{1}>, ..., <\varphi_{n}, \psi_{n}>\}$. We construct the
pure term $t$ such that $t$ is typable in $\Gamma$ if and only if the Post 
problem has a solution.
\begin{tabbing}
\=aaaaaaaaaaaaaaaaaaa\=aaaaaaaaaaaaaa \=\kill
\> \>$t = \lambda f.\lambda g.\lambda h.(f$ \> $(g~a~...~a)$\\
\> \>               \> $(h~(g~\tilde{\varphi_{1}}~...~\tilde{\varphi_{n}}))$\\
\> \>               \> $(h~(g~\tilde{\psi_{1}}~...~\tilde{\psi_{n}}))$\\
\> \>               \> $(F~c~(g~\lambda y.y~...~\lambda y.y))$\\
\> \>               \> $(F~d~(g~\lambda y.d~...~\lambda y.d)))$\\
\end{tabbing}
Assume this term is typable and call $\alpha$ the type of $g$.
The term $(g~a~...~a)$ is well-typed and is an object in $\Gamma \Delta$ so 
$$\alpha \cong \Pi x_{1}:T \ra T....\Pi x_{n}:T \ra T.
(\beta~x_{1}~...~x_{n})$$
where $\beta$ is a term of type $(T \ra T) \ra ... \ra (T \ra T) \ra Type$
in $\Gamma \Delta$.

Then all the variables $y$ bound in the terms $\tilde{\varphi_{i}}$, 
$\tilde{\psi_{i}}$, $\lambda y.y$ and $\lambda y.d$ have type $T$.
The term $(g~\hat{\varphi_{1}}~...~\hat{\varphi_{n}})$ has the type 
$(\beta~\hat{\varphi_{1}}~...~\hat{\varphi_{n}})$, so from the well-typedness
of the term $(h~(g~\hat{\varphi_{1}}~...~\hat{\varphi_{n}}))$
we get that the type of the variable $h$ has the form $\Pi x:\gamma.\gamma'$ 
and
$$\gamma \cong (\beta~\hat{\varphi_{1}}~...~\hat{\varphi_{n}})$$ 
in the same way, from the well-typedness of the term 
$(h~(g~\hat{\psi_{1}}~...~\hat{\psi_{n}}))$
we get
$$\gamma \cong (\beta~\hat{\psi_{1}}~...~\hat{\psi_{n}})$$
so
$$(\beta~\hat{\varphi_{1}}~...~\hat{\varphi_{n}}) \cong
                                    (\beta~\hat{\psi_{1}}~...~\hat{\psi_{n}})$$
From the well-typedness of the term 
$(F~c~(g~\lambda y:T.y~...~\lambda y:T.y))$ we get
$$(\beta~\lambda y:T.y~...~\lambda y:T.y) \cong (P~c)$$
At last from the the well-typedness of the term 
$(F~d~(g~\lambda y:T.d~...~\lambda y:T.d))$
we get
$$(\beta~\lambda y:T.d~...~\lambda y:T.d) \cong (P~d)$$
Since the term $\beta$ has type $(T \ra T) \ra ... \ra (T \ra T) \ra Type$, 
the head symbol of the normal form of the term $\beta$ cannot be a top 
variable of $\beta$, so it is the variable $P$ and we have 
$$\beta \cong 
 \lambda x_{1}:T \ra T....\lambda x_{n}:T \ra T.(P~(\delta~x_{1}~...~x_{n}))$$
For some term $\delta$ of type $(T \ra T) \ra ... \ra (T \ra T) \ra T$.
We get
$$(\delta~\hat{\varphi_{1}}~...~\hat{\varphi_{n}}) 
                         \cong (\delta~\hat{\psi_{1}}~...~\hat{\psi_{n}})$$
$$(\delta~\lambda y:T.y~...~\lambda y:T.y) \cong c$$
$$(\delta~\lambda y:T.d~...~\lambda y:T.d) \cong d$$
The second equality shows that the normal form of the term $\delta$ has the 
form
$$\lambda x_{1}:T \ra T....\lambda x_{n}:T \ra T.
(x_{i_{1}}~(...(x_{i_{p}}~c)...))$$
for some sequence $i_{1}, ..., i_{p}$.
The third equality shows that $p > 0$ and the first one that
$$(\hat{\varphi_{i_{1}}}~(...(\hat{\varphi_{i_{p}}}~c)...)) \cong
                          (\hat{\psi_{i_{1}}}~(...(\hat{\psi_{i_{p}}}~c)...))$$
so the sequence $i_{1}, ..., i_{p}$ is a solution to the Post problem.

Conversely assume that the Post problem has a solution $i_{1}, ..., i_{p}$, 
then by giving the following types to the variables $f$, $g$ and $h$
$$f:(P~(a~(...(a~c)...))) \ra T \ra T \ra T \ra T \ra T$$
$$g:\Pi x_{1}:T \ra T....\Pi x_{n}:T \ra T.
(P~(x_{i_{1}}~(...(x_{i_{p}}~c)...)))$$
$$h:(P~(\hat{\varphi_{i_{1}}}~(...(\hat{\varphi_{i_{p}}}~c)...))) \ra T$$
and the type $T$ to all the other variables of the term $t$, we get a term 
$t'$ well-typed in $\Gamma$, which is an object and such that $t = |t'|$.
\end{proof}

\begin{remark}
Along the way, we have proved that in the simply typed $\lambda$-calculus, 
the unification problems of the form
$$(f~t_{1}~...~t_{n}) = (f~t'_{1}~...~t'_{n})$$
$$(f~u_{1}~...~u_{n}) = u'$$
$$(f~v_{1}~...~v_{n}) = v'$$
where $t_{i}, t'_{i}, u_{i}, u', v_{i}, v'$ are closed terms and $f$ a 
third order variable are undecidable.

It is decidable if each of these equations has a solution or not 
(since the first one is flexible-flexible \cite{Huet75} \cite{Huet76} and the 
others third order matching problems \cite{lics}), but it is undecidable 
whether or not they have a solution {\it in common}.
If the variable $f$ is second order the problems of this form
are decidable since the second order matching algorithm \cite{Huet76} 
\cite{HueLan} produces a finite complete set of closed solutions.
\end{remark}

\section*{Acknowledgements}

The author thanks Frank Pfenning for many stimulating and helpful discussions 
on this problem and Pawel Urzyczyn for his careful reading of a previous draft
of this paper.

\end{document}